\documentclass[a4paper,USenglish]{lipics-v2021}
\hideLIPIcs
\nolinenumbers

\usepackage{todonotes}

\usepackage{xspace}
\def\epsilon{\varepsilon}
\newcommand{\Beps}[0]{B^\epsilon}
\newcommand{\logB}{\log_{B}}
\newcommand{\logMB}{\log_{M/B}}
\newcommand{\IO}[0]{I/O\xspace}
\newcommand{\IOs}[0]{I/Os\xspace}

\newcommand{\BepsTree}[0]{$\Beps$\nobreakdash-tree\xspace}
\newcommand{\BepsTrees}[0]{$\Beps$\nobreakdash-trees\xspace}

\newcommand{\vbar}[0]{\overline{v}}
\newcommand{\Nbar}[0]{\overline{N}}
\newcommand{\versionSet}[1]{\mathcal{S}_{#1}}
\newcommand{\Nv}[0]{N_v}
\newcommand{\versionCurrent}{v_c}

\newcommand{\Operation}[1]{{\normalfont\textsc{#1}}\xspace}
\newcommand{\Insert}{\Operation{Insert}}
\newcommand{\Delete}{\Operation{Delete}}
\newcommand{\Range}{\Operation{Range}}
\newcommand{\Search}{\Operation{Search}}

\newcommand{\lrCurlyBrackets}[1]{\left\{ #1 \right\}}
\newcommand{\lrParents}[1]{\left( #1 \right)}

\newcommand{\lrCeil}[1]{\left\lceil #1 \right\rceil}

\newcommand{\Oh}[1]{\mathcal{O}\!\lrParents{#1}}
\newcommand{\OhMega}[1]{\Omega\!\lrParents{#1}}
\newcommand{\OhTheta}[1]{\Theta\!\lrParents{#1}}
\newcommand{\sort}[1]{\mathit{sort}\!\lrParents{#1}}

\usepackage{xcolor}
\usepackage{tikz}
\usetikzlibrary{calc,math,decorations.pathreplacing,positioning}

\newcommand{\CLOSEDRECTANGLE}[4]{% #1 lower left, #2 width, #3 height, #4 text
  \begin{scope}[shift={(#1)}]
    \draw[very thick] (0,0) rectangle (#2,#3);
    \ifstrempty{#4} {} {
        \node (corner) at (#2,0) {};
        \node[inner sep=1, align=left, above left=-.5pt and -.5pt of corner, fill=white] {#4};
    }
  \end{scope}}

\newcommand{\OPENRECTANGLE}[4]{% #1 lower left, #2 width, #3 height, #4 text
  \begin{scope}[shift={(#1)}]
    \draw[white, fill=gray!10] (0,0) rectangle (#2,#3);
    \draw[very thick] (0,#3) -- (0,0) -- (#2,0) -- (#2,#3);
    \ifstrempty{#4} {} {
        \node (corner) at (#2,0) {};
        \node[inner sep=1, align=left, above left=-.5pt and -.5pt of corner, fill=none] {#4};
    }
  \end{scope}}

\newcommand{\OPENSEGMENTS}[3][black]{%
  \foreach \x/\ystart in {#3} {
    \draw[#1] (\x,\ystart) -- (\x,#2);
    \draw[#1, fill=#1] (\x,\ystart) circle (0.075);
  }}

\newcommand{\BUFFEREDOPENSEGMENTS}[0]{\OPENSEGMENTS[lightgray]}

\newcommand{\CLOSEDSEGMENTS}[2][black]{
  \foreach \x/\ystart/\yend in {#2} {
    \draw[#1] (\x,\ystart) -- (\x,\yend-0.075);
    \draw[#1, fill=#1] (\x,\ystart) circle (0.075);
    \draw[#1, fill=none] (\x,\yend) circle (0.075);
  }}

\newcommand{\BUFFEREDCLOSEDSEGMENTS}[0]{\CLOSEDSEGMENTS[lightgray]}

\newcommand{\FADESEGMENTS}[1]{
  \foreach \x/\ystart/\yend in {#1} {
    \coordinate (A) at (\x,\ystart);
    \tikzmath{\z = \ystart + ((\yend - \ystart) / 4 * 3);}
    \coordinate (B) at (\x,\z);
    \foreach \i in {0, 1, ..., 9} {
      \tikzmath{int \j; \j = \i + 1; \k = (\i + .5) / 10 * 100;}
      \draw[color=black!\k!lightgray] ($(B)!\i/10!(A)$) -- ($(B)!\j/10!(A)$);
    }
    \draw[lightgray] (B) -- (\x,\yend-0.075);
    \draw[black, fill=black] (A) circle (0.075);
    \draw[lightgray, fill=none] (\x,\yend) circle (0.075);
  }}

\title{Buffered Partially-Persistent External-Memory Search Trees}

\author{Gerth St{\o}lting Brodal}{Aarhus University, Denmark}{gerth@cs.au.dk}{https://orcid.org/0000-0001-9054-915X}{}

\author{Casper Moldrup Rysgaard}{Aarhus University, Denmark}{rysgaard@cs.au.dk}{https://orcid.org/0000-0002-3989-123X}{}

\author{Rolf Svenning}{Aarhus University, Denmark}{rolfsvenning@cs.au.dk}{https://orcid.org/0000-0002-9903-4651}{}

\funding{All authors are funded by Independent Research Fund Denmark, grant~9131-00113B.}

\authorrunning{G.\,S. Brodal, C.\,M. Rysgaard and R. Svenning}

\Copyright{Gerth St{\o}lting Brodal, Casper Moldrup Rysgaard, Rolf Svenning}

\keywords{B-tree, buffered updates, partial persistence, external memory}

\ccsdesc{Theory of computation~Data structures design and analysis}

\EventEditors{John Q. Open and Joan R. Access}
\EventNoEds{2}
\EventLongTitle{42nd Conference on Very Important Topics (CVIT 2016)}
\EventShortTitle{CVIT 2016}
\EventAcronym{CVIT}
\EventYear{2016}
\EventDate{December 24--27, 2016}
\EventLocation{Little Whinging, United Kingdom}
\EventLogo{}
\SeriesVolume{42}
\ArticleNo{23}

\begin{document}

\maketitle

\begin{abstract}
    We present an optimal partially-persistent external-memory search tree with amortized I/O bounds matching those achieved by the non-persistent $B^{\varepsilon}$-tree by Brodal and Fagerberg [SODA 2003].
    In a partially-persistent data structure each update creates a new version of the data structure, where all past versions can be queried, but only the current version can be updated. 
    All operations should be efficient with respect to the size $N_v$ of the accessed version $v$. 
    For any parameter~$0 < \varepsilon < 1$, our data structure supports insertions and deletions in amortized $\mathcal{O}\!\left(\frac{1}{\varepsilon B^{1 - \varepsilon}} \log_B N_v\right)$ I/Os, where $B$ is the external-memory block size.
    It also supports successor and range reporting queries in amortized $\mathcal{O}\!\left(\frac{1}{\varepsilon} \log_B N_v + K / B\right)$ I/Os, where $K$ is the number of values reported. 
    The space usage of the data structure is linear in the total number of updates. 
    We make the standard and minimal assumption that the internal memory has size $M \geq 2B$. 
    The previous state-of-the-art external-memory partially-persistent search tree by Arge, Danner and Teh [JEA 2003] supports all operations in worst-case $\mathcal{O}\!\left(\log_B N_v + K / B\right)$ I/Os, matching the bounds achieved by the classical B-tree by Bayer and McCreight [Acta Informatica~1972]. 
    Our data structure successfully combines buffering updates with partial persistence. 
    The I/O bounds can also be achieved in the worst-case sense, by slightly modifying our data structure and under the requirement that the memory size $M = \Omega\!\left(B^{1-\varepsilon}\log_2 (\max_v N_v)\right)$.
    For updates, where the I/O bound is~$o(1)$, we assume that the I/Os are performed evenly spread out among the updates (by performing buffer-overflows incrementally).
    The worst-case result slightly improves the memory requirement over the previous ephemeral external-memory dictionary by
    Das, Iacono, and Nekrich (ISAAC 2022), who achieved matching worst-case I/O bounds but required $M = \Omega\!\left(B\log_B N\right)$.
\end{abstract}

% --------------------------------------------------------------------
\section{Introduction}
\label{sec:introduction}
% --------------------------------------------------------------------

Developing data structures for storing a set of values from a totally ordered set subject to insertions, deletions, successor and predecessor queries, and range reporting queries is a fundamental problem in computer science. 
The classical solution in external-memory is the B-tree by Bayer and McCreight~\cite{BayerMcCreight72} which supports all the operations in worst-case $\Oh{\logB N + N/K}$ \IOs, where $N$ is the current size of the set, $K$ is the number of reported values, and $B$ is the external-memory block size. 
While the B-tree achieves the optimal number of \IOs for queries, for any $0 < \epsilon < 1$, the \BepsTree by Brodal and Fagerberg~\cite{BrodalFagerberg03lower} significantly improves update efficiency by attaching buffers to the internal nodes of a B-tree. 
This design supports updates with amortized $\Oh{\frac{1}{\epsilon B^{1 - \epsilon}} \logB{N}}$ \IOs. 
The $\frac{1}{\epsilon}B^{1 - \epsilon}$ factor improvement over traditional B-trees is significant when considering typical parameters of, e.g., $\epsilon = 1/2$ and $B = 1000$~\cite{Arge07} and the \BepsTree has found important applications in high-performance industry software such as TokuDB~\cite{BenderFJJKOYZ15} and BetrFS~\cite{JannenEtal15}. 

Although the \BepsTree optimizes update efficiency, it is \emph{ephemeral}, like most dynamic data structures, meaning that each update overwrites the previous version, and only the current version can be queried. 
In many applications, maintaining access to past versions is beneficial or even essential. 
A \emph{persistent} data structure supports such accesses, and in their seminal 1989 paper, Driscoll, Sarnak, Sleator, and Tarjan introduced general techniques for making ephemeral data structures persistent~\cite{DriscollSarnakTarjan89}. 
In particular, a \emph{partially-persistent} data structure supports queries in all past versions of the data structure but only the current version can be updated. 
Multiple authors have adapted these techniques to the external-memory model, developing partially-persistent B-trees that support updates and queries in worst-case $\Oh{\logB N_v + K / B}$ \IOs, where $N_v$ is size of the accessed version $v$, matching the performance of classical B-trees~\cite{ArgeDannerTeh03, BeckerGOSW96, VarmanVerma07}. 

In this paper, we present the first buffered partially-persistent external-memory search tree that retains the optimal update and query performance of the (ephemeral) \BepsTree. 
Our approach combines buffering techniques, which are essential for efficient updates in external memory, with a geometric view of persistence. 

% --------------------------------------------------------------------
\subsection{The External-Memory Model}
% --------------------------------------------------------------------

For problems on massive amounts of data that do not fit in internal memory, the standard model of computation is the \IO-model by Aggarwal and Vitter~\cite{AggarwalVitter88}. 
In this model, all computation occurs in an internal memory of size $M$, while an infinite external memory is used for storage. 
Data is transferred between internal and external memory in blocks of $B$ consecutive elements, with each transfer counting as an \IO. 
The \IO complexity of an algorithm is defined as the total number of \IOs it performs, and the space usage is the maximum number of external-memory blocks used at any given time. 
The only operation we allow on stored values are comparisons and we follow the standard assumption that the parameters $B \ge 2$ and $M \geq 2B$.
Aggarwal and Vitter proved that the optimal bound for sorting in external memory is $\sort{N} = \OhTheta{\frac{N}{B} \logMB \frac{N}{B}}$ I/Os~\cite{AggarwalVitter88}.
An algorithm is called \emph{cache oblivious} if it is designed without explicit knowledge of $B$ and $M$ but is still analyzed in the \IO model for arbitrary values of these parameters, assuming an optimal offline cache replacement strategy~\cite{FigoLeisersonProkopRamachandran12}. 
Some authors make stronger assumptions on the size of the internal memory, such as the \emph{tall-cache assumption} $M \geq B^{1 + \delta}$, for some constant~$\delta > 0$. 
For cache-oblivious algorithms, a tall-cache assumption is necessary to achieve optimal comparison-based external-memory sorting~\cite{BrodalFagerberg03}.

Being considerate of the \IO-behavior of algorithms can be crucial in practice, as demonstrated by Streaming B-trees~\cite{BenderFFFKN07}, the generation of massive graphs for the LFR benchmark~\cite{HamannMPTW18, LancichinettiFortunato09, LancichinettiFortunatoRadicchi08}, and the FlashAttention algorithm used in Transformer models~\cite{DaoFuErmonRudraRe22}.

% --------------------------------------------------------------------
\subsection{Interface of a Partially-Persistent Search Tree}
\label{sec:interface}
% --------------------------------------------------------------------

A partially-persistent search tree stores an ordered set of values supporting the below interface (in our examples we use integers, but our data structure works for any totally ordered set).
Each version is identified by a unique integer version identifier~$v$, with zero being the initial version and the \emph{current} version denoted by~$\versionCurrent$.
Further, we let $\versionSet{v}$ denote the set of values contained in version~$v$, and $\Nv$ the size of $\versionSet{v}$.
Initially $\versionCurrent = 0$ and $\versionSet{\versionCurrent} = \emptyset$.
Updates (insertions and deletions) can only be performed on the current set $\versionSet{\versionCurrent}$, and any update advances the current version identifier, i.e., each version of the set $\versionSet{v}$ only differs from the previous version $\versionSet{v-1}$ by at most a single value.
Queries can be performed on any version.

\begin{description}
    \item[$\Insert(x)$]
      Creates $\versionSet{\versionCurrent + 1} = \versionSet{\versionCurrent} \cup \{ x \}$, increments $\versionCurrent$, and returns $\versionCurrent$.
    \item[$\Delete(x)$]
      Creates $\versionSet{\versionCurrent + 1} = \versionSet{\versionCurrent} \setminus \{ x \}$, increments $\versionCurrent$, and returns $\versionCurrent$.
    \item[$\Range(v, x, y)$]
      Reports all values in $\versionSet{v} \cap [x,y]$ in increasing order.
    \item[$\Search(v, x)$]
      Returns the successor of $x$ in $\versionSet{v}$, i.e., $\min \{ y \in \versionSet{v} \; | \; x \le y \}$.
\end{description}

% --------------------------------------------------------------------
\subsection{Previous Work}
% --------------------------------------------------------------------

\begin{table}[t]
    \centering
    \caption{Overview of results on the \IO complexity of ephemeral and partial persistent search trees.
    Results marked by ``am.'' hold amortized, and results marked by ``rand.'' are randomized and hold with high probability. All other results hold in worst case. The parameter
    $\epsilon$ must satisfy $0 < \epsilon < 1$.
    All results assume $M = \OhMega{B}$, further $\dag$ assumes $B = \OhMega{\log N}$ and $M = \OhMega{\max \{ B\log^{\OhTheta{1}} N, B^2 \}}$; $\ddag$ assumes $M = \OhMega{B \logB N}$; and $*$ assumes $M = \OhMega{B^{1-\epsilon}\log_2 (\max_v N_v)}$.
    For both queries and updates in \cite{BenderDFJK20, DasIaconoNekrich22}, we include the multiplicative dependency on $\frac{1}{\epsilon}$ (that can be omitted when treating $\epsilon$ as a constant), allowing, for example, setting $\epsilon = \frac{1}{\log_2 B}$. 
    All ephemeral results use space linear in $N$ and all partial persistence results use space linear in the total number of updates.}
    \label{tab:overview_of_results}
    {
    \newcommand{\lines}[2]{\begin{tabular}{@{}l@{}}#1\\[-0.5ex]\quad #2\end{tabular}}
    \renewcommand{\arraystretch}{1.2}
    \begin{tabular}{l@{}cc}
        & Range Query & Update \\ \hline
        \textbf{Ephemeral} \\
        Bayer and McCreigh~\cite{BayerMcCreight72}
        & $\Oh{\logB N + K/B}$ 
        & $\Oh{\logB N}$ \\
        Brodal and Fagerberg~\cite{BrodalFagerberg03lower}
        & $\Oh{\frac{1}{\epsilon}\logB N + K/B}$ am.
        & $\Oh{\frac{1}{\epsilon B^{1-\epsilon}}\logB N}$ am. \\
        \lines{Bender, Das, Farach-Colton,}{Johnson, and Kuszmaul$^\dag$~\cite{BenderDFJK20}}
        & $\Oh{\frac{1}{\epsilon}\logB N + K/B}$
        & $\Oh{\frac{1}{\epsilon B^{1-\epsilon}}\logB N}$ rand. \\
        Das, Iacono, and Nekrich$^\ddag$~\cite{DasIaconoNekrich22}
        & $\Oh{\frac{1}{\epsilon}\logB N + K/B}$ 
        & $\Oh{\frac{1}{\epsilon B^{1-\epsilon}}\logB N}$ \\[1ex] \hline
        \textbf{Partial Persistent} \\
        \lines{Becker, Gschwind, Ohler,}{Seeger, and Widmayer~\cite{BeckerGOSW96}}
        \smash{\raisebox{-4.25ex}[0pt]{$\left.\rule{0ex}{7.5ex}\right\}$}}
        & \multirow{3}{*}{$\Oh{\logB \Nv + K/B}$}
        & \multirow{3}{*}{$\Oh{\logB \Nv}$} \\
        Varman and Verma~\cite{VarmanVerma07} \\
        Arge, Danner, and Teh~\cite{ArgeDannerTeh03}
        \\
        \emph{This paper (Theorem~\ref{thm:main})}
        & $\Oh{\frac{1}{\epsilon}\logB \Nv + K/B}$ am.
        & $\Oh{\frac{1}{\epsilon B^{1-\epsilon}}\logB \Nv}$ am.
        \\
        \emph{This paper (Theorem~\ref{thm:worst-case-large-cache})$^*$}
        & $\Oh{\frac{1}{\epsilon}\logB \Nv + K/B}$
        & $\Oh{\frac{1}{\epsilon B^{1-\epsilon}}\logB \Nv}$ 
        \\
        \emph{This paper (Theorem~\ref{thm:worst-case-small-cache})}
        & $\Oh{\frac{1}{\epsilon}\logB \Nv + \gamma + K/B}$
        & $\Oh{\frac{1}{B^{1-\epsilon}} \lrParents{ \frac{1}{\epsilon} \logB \Nv + \gamma }}$
        \\
        & \multicolumn{2}{c}{$\gamma = \sort{B^{1-\epsilon} \log_2 \Nv} = \frac{\log_2 \Nv}{B^\epsilon} \logMB \frac{\log_2 \Nv}{B^\epsilon}$}
        \\
        \hline
    \end{tabular} }
\end{table}

In internal memory, the \emph{fat node} and \emph{node copying} techniques can make any ephemeral linked data structure partially-persistent with constant overhead in both time and space, as long as the in-degree of each node in the ephemeral structure is constant~\cite{DriscollSarnakTarjan89}. 
Becker, Gschwind, Ohler, Seeger, and Widmayer~\cite{BeckerGOSW96} and Varman and Verma~\cite{VarmanVerma07} adapted these techniques to B-trees in external-memory.
An elegant application of partial persistence appears in the design of linear space planar point location data structures~\cite{SarnakTarjan86}. 
In this setting, the underlying set consists of segments which are partially ordered (only a pair of segments intersected by a vertical line can be compared). 
To adapt this approach to the external-memory setting, Arge, Danner, and Teh strengthened the partially-persistent B-tree to require only a total order on values alive at any given version, leading to a static external-memory point-location structure~\cite{ArgeDannerTeh03}. 

A different approach to persistence is to interpret it geometrically (Figure~\ref{fig:geometry}), modeling it as a data structure problem on a dynamic set of vertical (or horizontal) segments. 
Kolovson and Stonebraker explored this perspective~\cite{KolovsonStonebraker89}, though their reliance on R-trees led to poor performance guarantees~\cite{Guttman84}. 
More recently, Brodal, Rysgaard, and Svenning~\cite{BrodalRysgaardSvenning23} leveraged this geometric approach to develop \emph{fully persistent} B-trees, which allow both queries and modifications to all past versions in $\Oh{\logB{N_v}}$ \IOs. 
In a fully persistent data structure, updating a version corresponds to cloning it and then applying the modification to the newly cloned version, ensuring that existing versions remain unaffected. 
Such behavior contrasts with \emph{retroactive data structures}~\cite{DemaineIaconoLangerman07}, where updates recursively propagate to cloned versions.

Concurrently with the work on persistent data structures in external-memory, there were significant improvements to external-memory data structures by leveraging buffering techniques to always process multiple updates and/or queries together. 
These include the Buffer Tree by Arge~\cite{Arge03} which can form the basis for external-memory sorting, priority queues and batched dynamic algorithms~\cite{EdelsbrunnerOvermars1985} in amortized $\Oh{\frac{1}{B} \logMB \frac{N}{B}}$ \IOs per operation. 
For a batched operation the answer might not be immediately returned, which is often sufficient, e.g., in many geometric plane-sweep algorithms where only the end result matters. 
For standard (non-batched) data structures, a line of work has investigated the update-query trade-off, beginning with the Buffered Repository Tree~\cite{BuchsbaumGVW00} performing updates in amortized $\Oh{\frac{1}{B} \logB{N}}$ \IOs and queries in $\Oh{\log_2{N}}$ \IOs. 
This was later generalized by the \BepsTree{}~\cite{BrodalFagerberg03lower} which for $\epsilon \approx 0$ corresponds to the Buffered Repository Tree and for $\epsilon \approx 1$ to the standard B-tree. 
The amortized performance of the \BepsTree was improved to high-probability~\cite{BenderDFJK20} and worst-case~\cite{DasIaconoNekrich22} \IO bounds using stronger assumptions on the size of $B$ and $M$ (see Table~\ref{tab:overview_of_results}). 

% --------------------------------------------------------------------
\subsection{Contribution}
% --------------------------------------------------------------------

Combining the two lines of research on persistence and buffered data structures has remained an open challenge for the past 20 years, likely due to their seemingly conflicting principles. 
Persistence requires maintaining access to past versions without affecting their structure, while buffers essentially hold updates to past versions before applying them. 
Our work demonstrates that that these two ideas can be effectively unified by developing partially-persistent external-memory search trees that achieve bounds matching those of ephemeral \BepsTrees. 

\begin{theorem}
\label{thm:main}
    Given any parameter $0 < \epsilon < 1$ and $M \geq 2B$, there exist partially-persistent external-memory search trees over any totally ordered set, that support \Insert and \Delete in amortized $\Oh{\frac{1}{\epsilon B^{1 - \epsilon}} \logB{N_v}}$ \IOs, \Search in amortized $\Oh{\frac{1}{\epsilon}\logB{N_v}}$ \IOs, and \Range in amortized $\Oh{\frac{1}{\epsilon}\logB{N_v} + K / B}$ \IOs. 
    Here $\Nv$ denotes the number of values contained in version~$v$, and $K$ the number of values reported by \Range.
    The space usage is linear in the total number of updates.
\end{theorem}

The query \Search can trivially also answer a member query ``$x\in \versionSet{v}$?'' by checking if $\Search(v, x)$ returns $x$.
Our data structure can further also support predecessor queries instead of successor queries, as well as strict predecessor and successor queries, i.e., the returned value should be strictly smaller or larger than the query value $x$.
The structure can also handle the case when $\versionSet{0} \ne \emptyset$, where the initial structure can be constructed using $\Oh{1+|\versionSet{0}|/B}$ \IOs
(essentially this is Section~\ref{sec:global_rebuild}, where a structure is constructed for a given set).
Our data structure is stated as maintaining a set of values, but it can easily be extended to support dictionaries storing key-value pairs 
(each segment in Figure~\ref{fig:geometry} and Figure~\ref{fig:geometric_partitioning} now stores a key-value pair, where the first axis now are keys; 
changing the value for key~$x$ at version $v$ starts a new vertical segment at $(x, v)$ with the new value). 

In Section~\ref{sec:worst-case} we describe how to convert the amortized \IO bounds of Theorem~\ref{thm:main} to worst-case bounds under the assumption that 
$M = \OhMega{B^{1-\epsilon} \log_2 (\max_v {N_v})}$ (Theorem~\ref{thm:worst-case-large-cache}),
a weaker or equal assumption on the memory size than used in \cite{BenderDFJK20} and \cite{DasIaconoNekrich22} for high-probability and worst-case bounds for \BepsTrees, respectively. 
Under the weakest assumption that $M  \geq 2B$,  we achieve the worst-case bounds in Theorem~\ref{thm:worst-case-small-cache} with an additional term of at most $\Oh{\sort{B^{1-\epsilon} \log_2 \Nv}}$ \IOs, 
where $B^{1-\epsilon} \log_2 \Nv$ is an upper bound on the number of buffered updates on a root-to-leaf path that should be flushed to the leaf, and $\sort{N}=\OhTheta{\frac{N}{B}\logMB \frac{N}{B}}$ denotes the number of \IOs to sort $N$ values~\cite{AggarwalVitter88}. For updates, where the \IO bound is~$o(1)$, we assume that the \IOs are performed evenly spread out among the updates.

\begin{theorem}
\label{thm:worst-case-large-cache}
    Given any parameter $0 < \epsilon < 1$ and $M = \OhMega{B^{1-\epsilon}\log_2 (\max_v N_v)}$, there exist partially-persistent external-memory search trees over any totally ordered set, that support \Insert and \Delete in worst-case $\Oh{\frac{1}{\epsilon B^{1 - \epsilon}} \logB{N_v}}$ \IOs, \Search in worst-case $\Oh{\frac{1}{\epsilon}\logB{N_v}}$ \IOs, and \Range in worst-case $\Oh{\frac{1}{\epsilon}\logB{N_v} + K / B}$ \IOs. 
    Here $\Nv$ denotes the number of values contained in version $v$, and $K$ the number of values reported by \Range.
    The space usage is linear in the total number of updates.
\end{theorem}

\begin{theorem}
\label{thm:worst-case-small-cache}
    Given any parameter $0 < \epsilon < 1$ and $M \geq 2B$, there exist partially-persistent external-memory search trees over any totally ordered set, that support \Insert and \Delete in worst-case $\Oh{\frac{1}{B^{1-\epsilon}} \lrParents{ \frac{1}{\epsilon} \logB \Nv + \gamma }}$ \IOs, \Search in worst-case $\Oh{\frac{1}{\epsilon}\logB{N_v} + \gamma}$ \IOs, and \Range in worst-case $\Oh{\frac{1}{\epsilon}\logB{N_v} + \gamma + K / B}$ \IOs,
    where $\gamma = \sort{B^{1-\epsilon} \log_2 \Nv}$. 
    Here $\Nv$ denotes the number of values contained in version $v$, and $K$ the number of values reported by \Range.
    The space usage is linear in the total number of updates.
\end{theorem}

Note that, for example, when $\Nv = 2^{\Oh{\Beps}}$ then $B^{1-\epsilon} \log_2 \Nv=\Oh{B}$ and $\gamma=\Oh{1}$ and the \IO bounds of Theorem~\ref{thm:worst-case-small-cache} match those of Theorem~\ref{thm:worst-case-large-cache}, with only the assumption $M\geq 2B$.
This observation can be further strengthened, as when $\gamma = \Oh{\frac{1}{\epsilon} \logB \Nv}$ the bounds \IO match similarly, which holds when $\Nv = 2^{B^\epsilon \lrParents{\frac{M}{B}}^{\Oh{\frac{B^\epsilon}{\epsilon \log_2 B}}}}$. 

\subparagraph*{Outline of Data Structure}

Previous work on partially-persistent search trees in external memory directly adapted the general pointer-based transformations for persistence~\cite{DriscollSarnakTarjan89}. 
In contrast, our approach embraces the geometric interpretation of partial persistence (see Figure~\ref{fig:geometry}) similar to that of~\cite{BrodalRysgaardSvenning23},
where the state of the data structure is embedded in a two-dimensional plane with values on the first axis and versions on the second axis. 
Under this interpretation, each update corresponds to the start or end of a vertical segment in the plane. 
Since partial persistence updates are
applied to the current version, it always affects the top of the plane. 
Successor and predecessor queries correspond to horizontal ray shooting to the right and left, respectively, and range queries to reporting the intersections between a horizontal query segment among vertical segments.

To efficiently update and query the geometric view, we partition the plane into rectangles, each containing $\OhTheta{\frac{1}{\epsilon} B \logB \Nbar}$ vertical segments in lexicographic order. 
For now we assume that all versions have size $\OhTheta{\Nbar}$, for a fixed $\Nbar$ (this assumption is lifted using global rebuilding, see Section~\ref{sec:global_rebuild}). 
In the geometric persistent view, a vertical segment crossing multiple rectangles is split into multiple smaller segments, one for each rectangle, and each smaller segment is inserted into one rectangle.

At a high level, our data structure is divided into two parts. The top part consists of all the open rectangles containing the current set $\versionSet{\versionCurrent}$, which may still be updated. 
The entry point of this data structure is a $B^{\epsilon}$-tree $T$ on the value axis to facilitate buffered updates and to find the relevant rectangle(s) for updates and queries. 
Since updates are buffered, the geometric view stored in the rectangles may be incomplete, since buffered updates (segment endpoints) will first be added to the rectangle when buffers are flushed.
The bottom part consists of all the finalized rectangles, i.e., rectangles which can be queried but not updated. 
The entry point to the bottom part is a data structure $P$ to find the relevant rectangle(s) for a query. 
This corresponds exactly to a point location problem and we implement $P$ as an external-memory adaption of the classical planar point location solution using partial persistence~\cite{ArgeDannerTeh03,SarnakTarjan86}, more specifically a B-tree with path copying during updates. 

% --------------------------------------------------------------------
\section{The Buffered Persistent Data Structure}
\label{sec:amortized-structure}
% --------------------------------------------------------------------

In this section, we describe our partially-persistent \BepsTree structure.
Versions are identified be the integers $0, 1, 2, \ldots$,  where $\versionCurrent$ denotes the identifier of the current version. We let $\versionSet{v}$ denote the set at version~$v$, where values are from some totally ordered set.
The initial set $\versionSet{0} = \emptyset$, and $\versionSet{v + 1} = \versionSet{v} \cup \{ x \}$ if the $v+1$'th update is $\Insert(x)$, and $\versionSet{v + 1} = \versionSet{v} \setminus \{ x \}$ if the $v+1$'th update is $\Delete(x)$.
Note that $\versionSet{v+1}=\versionSet{v}$ if the $(v+1)$'th update inserts a value already in~$\versionSet{v}$ or is deleting a value not in~$\versionSet{v}$.

% --------------------------------------------------------------------
\subsection{Geometric Interpretation of Partial Persistence}
\label{sec:persistence_geometry}
% --------------------------------------------------------------------

The problem has a natural geometric interpretation in a two dimensional space, with the first dimension representing the values and the second dimension representing the versions, see Figure~\ref{fig:geometry}.
On this two dimensional plane, a value $x$ existing in versions $[v, w[$, can be represented by the vertical line segment $\{x\} \times [v, w[$, i.e., $x$ is inserted in version $v$ and deleted in version $w$.
If $x \in \versionSet{\versionCurrent}$, then $w = +\infty$ ($x$ has not been deleted yet).

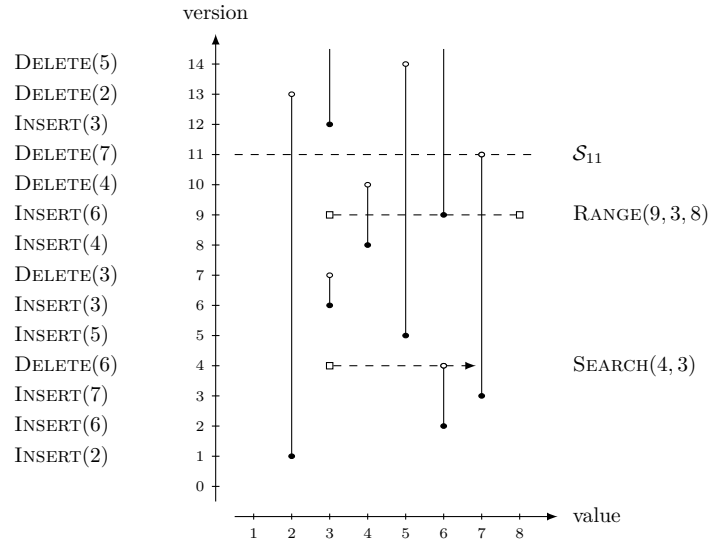
\begin{figure}
    \centering
    \begin{tikzpicture}[scale=0.5, every node/.style={scale=0.8}, yscale=0.8]
        % Query
        \draw[dashed] (0.5,11) -- (8.5,11);
        \node[inner sep=1.5, rectangle, draw] at (3,9) (x) {};
        \node[inner sep=1.5, rectangle, draw] at (8,9) (y) {};
        \node[inner sep=1.5, rectangle, draw] at (3,4) (z) {};
        \coordinate (z') at (6.85,4) {};
        \draw[dashed] (x) -- (y);
        \draw[-latex,dashed] (z) -- (z');
        \node[label={right:$\versionSet{11}$}] at (9,11) {};
        \node[label={right:$\Range(9, 3, 8)$}] at (9,9) {};
        \node[label={right:$\Search(4, 3)$}] at (9,4) {};

        % X-axis
        \foreach \x in {1, 2, ..., 8} {
          \draw (\x, -1.1) -- +(0, 0.2);
          \node[label={below:\scriptsize\x}] at (\x, -0.9) {};
        }
        \draw[-latex] (0.5,-1) -- (9,-1);
        \node[label={[yshift=1pt]right:value}] at (9,-1) {};

        % Y-axis
        \foreach \y in {0, 1, ..., 14} {
          \draw (-0.1, \y) -- +(0.2, 0);
          \node[label={left:\scriptsize\y}] at (0.1, \y) {};
        }
        \draw[-latex] (0,-0.5) -- (0,15);
        \node[label={above:version}] at (0,15) {};

        % Segments
        \CLOSEDSEGMENTS{ 2/1/13, 3/6/7, 4/8/10, 5/5/14, 6/2/4, 7/3/11 }
        \OPENSEGMENTS{14.5}{ 3/12, 6/9 }

        \draw[black, fill=white] (6, 4) circle (0.075);
        \draw[black, fill=white] (7, 11) circle (0.075);

        % Update list
        \foreach \instr [count=\y] in {\Insert(2), \Insert(6), \Insert(7), \Delete(6), \Insert(5), \Insert(3), \Delete(3), \Insert(4), \Insert(6), \Delete(4), \Delete(7), \Insert(3), \Delete(2), \Delete(5)} {
            \node[anchor=west] at (-5.5, \y) {$\instr$};
        }
    \end{tikzpicture}
    
    \caption{(Left) A list of updates performed on an initially empty set and (Right) the geometric interpretation of the updates. Vertical lines illustrate the interval of versions containing a value. Note that the value $4$ is contained in versions $[8,10[$, i.e., versions $8$ and $9$, 
    whereas the value~$3$ is contained in version $6$ and versions $[12, \infty[$.
    The topmost dashed line shows that version $11$ of the set is $\versionSet{11}=\{ 2, 5, 6 \}$, the dashed line segment at version $9$ shows that the result of the query $\Range(9, 3, 8)$ is $\{ 4, 5, 6, 7 \}$, and the bottommost dashed arrow shows that the result of the successor search $\Search(4, 3)$ is $7$.}
    \label{fig:geometry}
\end{figure}

% --------------------------------------------------------------------
\subsection{Partitioning the Plane into Rectangles}
\label{sec:partitioning}
% --------------------------------------------------------------------

\begin{figure}[t!]
    \centering

    \newcommand{\AddBuffer}[2]{
        \draw ($(#1.east) + (0.25, -0.2)$) rectangle +(1.4, 0.4);
        \foreach[parse=true] [count=\i] \f in {#2}
            \node[circle, draw, inner sep=1.5, lightgray, fill=\f] at ($(#1.east) + (0.25 + \i * 0.2, 0)$) {};
    }
    
    \begin{tikzpicture}[every node/.style={scale=0.9}]
        % Rectangles
        \CLOSEDRECTANGLE{0,0}{3}{4}{\small $r_1$}
        \OPENRECTANGLE{0,4}{1.5}{3}{}
        \OPENRECTANGLE{1.5,4}{1.5}{3}{}
        \CLOSEDRECTANGLE{3,0}{3.5}{2.5}{\small $r_2$}
        \CLOSEDRECTANGLE{3,2.5}{2}{2.5}{}
        \OPENRECTANGLE{3,5}{2}{2}{\small $r_5$}
        \CLOSEDRECTANGLE{6.5,0}{1}{1}{}
        \CLOSEDRECTANGLE{7.5,0}{1.5}{1}{}
        \CLOSEDRECTANGLE{6.5,1}{2.5}{1.5}{\small $r_3$}
        \OPENRECTANGLE{5,2.5}{4}{4.5}{\small $r_6$}
        \OPENRECTANGLE{9,0}{1.5}{7}{\small $r_4$}

        % Segments
        \CLOSEDSEGMENTS{
            0.5/0.3/2.1,
            0.75/3.5/4,
            1.2/2.4/4,
            1/0.75/1.7,
            2/0/4, 2/4/4.6,
            2.4/0/0.8,
            3.5/0/2.5, 3.5/2.5/5, 3.5/5/5.6,
            4/0.9/1.8, 4/3.2/3.9,
            4.2/4.5/5,
            4.4/2/2.5, 4.4/2.5/3.7,
            5.4/0/2.5, 5.4/2.5/2.9,
            6.3/3/4.3,
            6.8/0/1,
            7.2/0.5/1, 7.2/1/2.5, 7.2/2.5/3.4,
            7.9/0/0.7,
            8.3/0/1, 8.3/1/1.9
        }
        \OPENSEGMENTS{6.9}{
            0.75/4,
            2.5/4,
            4.2/5,
            8/2.5,
            9.95/0
        }
        \FADESEGMENTS{
            1.2/4/5.5,
            7/4.1/5.2,
            9.4/0.6/4.9
        }
        \BUFFEREDCLOSEDSEGMENTS{
            6.5/5.9/6.5,
            7.5/4.7/6
        }
        \BUFFEREDOPENSEGMENTS{6.9}{
            0.25/5,
            4.6/6.3,
            8.6/5.3,
            9.6/6.2
        }

        \draw[black, fill=white] (0.75,4) circle (0.075);
        \draw[black, fill=white] (6.8, 1) circle (0.075);
        \draw[black, fill=white] (4.2, 5) circle (0.075);

        % Queries
        \node[inner sep=1.5, rectangle, draw] at (1.6,1.5) (x) {};
        \node[inner sep=1.5, rectangle, draw] at (9.7,1.5) (y) {};
        \node[inner sep=1.5, rectangle, draw] at (4.5,5.55) (z) {};
        \coordinate (z') at (7.45,5.55) {};
        
        \draw[dashed] (x) -- (y);
        \draw[-latex,dashed] (z) -- (z');
        
        \node[label={below:$x$}] at (1.6,0) {};
        \node[label={below:$y$}] at (9.7,0) {};
        \node[label={below:$z$}] at (4.5,0) {};
        \node[label={left:$v$}] at (0,1.5) {};
        \node[label={left:$w$}] at (0,5.55) {};

        % Axis
        \draw[|-|] (0,-1) -- +(10.5,0) node[midway, fill=white, yshift=1pt] {value};
        \draw[|-latex] (-1,0) -- +(0,7) node[midway, fill=white, rotate=90, yshift=1pt] {version};

        % Upper tree
        \node[circle, draw, inner sep=6] (A) at (5.125, 11) {};
        \node[circle, draw, inner sep=6] (B) at (2.5, 9) {};
        \node[circle, draw, inner sep=6] (C) at (7.75, 9) {};
        \AddBuffer{A}{lightgray, white, white, lightgray}
        \AddBuffer{B}{lightgray, white}
        \AddBuffer{C}{lightgray, white, lightgray, lightgray, white}
        \draw (A) -- (B);
        \draw (A) -- (C);
        \draw (B) -- (0.75, 7.5);
        \draw (B) -- (2.25, 7.5);
        \draw (B) -- (4.0, 7.5);
        \draw (C) -- (7.0, 7.5);
        \draw (C) -- (9.75, 7.5);

        % Parameters
        \node at (3.9,8.4) {$< 2\Delta$};
        \node at (9,9.5) {$< 2\Delta F$};
        \draw[decorate,decoration={brace,amplitude=10pt}] (1,8.5) -- node[yshift=5ex,sloped,rotate=90] {$H$} +(0,3);
        \draw[dashed] (1.5,8.5) arc (230:310:1.5);
    \end{tikzpicture}
    
    \caption{(Top) A \BepsTree of the open rectangles and (Bottom) the geometric interpretation of the updates, split into multiple rectangles, where gray rectangles are open rectangles.
    The black endpoints represent updates present in the rectangle and the gray endpoints represent buffered updates present in the buffers at the internal nodes of the \BepsTree.
    A query $\Range(v, x, y)$ is represented as the dashed line between two square endpoints, spanning rectangles $r_1$, $r_2$, $r_3$, and~$r_4$, and a successor query $\Search(w, z)$ is represented as the dashed arrow from a square endpoint, spanning rectangles $r_5$ and $r_6$.
    Black dots on vertical segments correspond to the upper endpoint of the segment in the rectangle below and the lower endpoint of the segment in the rectangle above.}
    \label{fig:geometric_partitioning}
\end{figure}
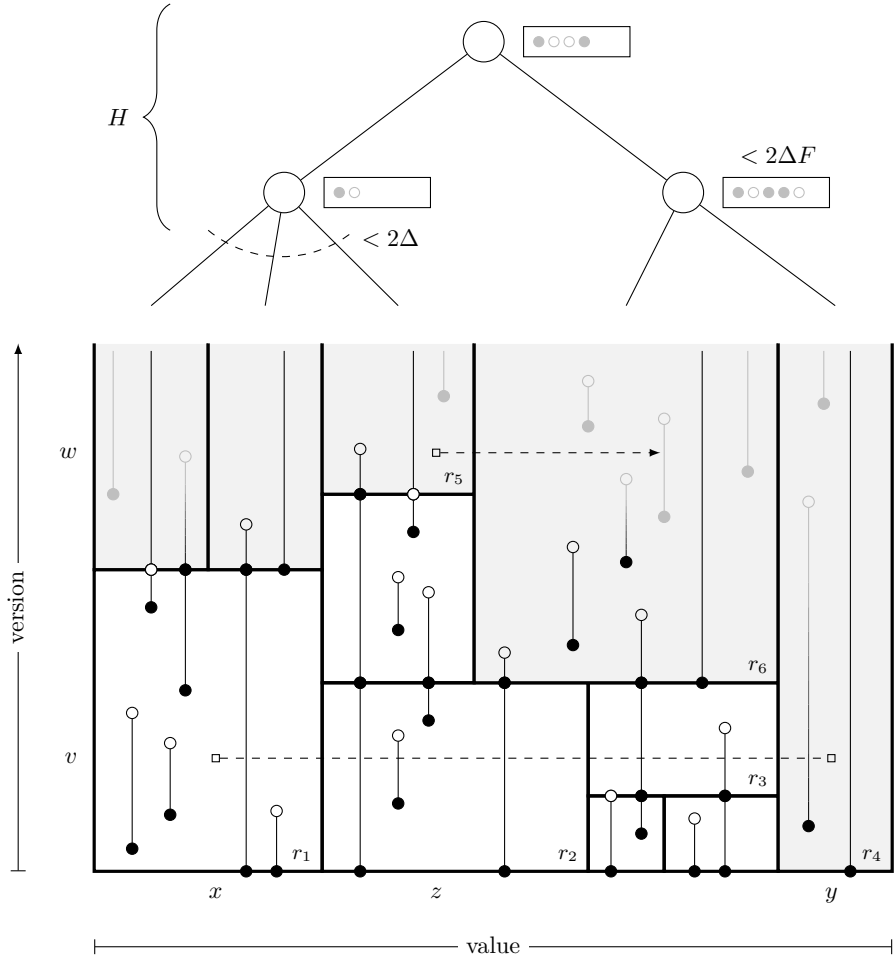

We consider the sequence of versions partitioned into intervals 
$[v_0,v_1[, \ldots, [v_{k-1}, v_k[, [v_k,\infty[$, 
for some versions $0=v_0 < v_1 < \cdots < v_k \leq v_c$. 
In Section~\ref{sec:global_rebuild} we show how to maintain the version intervals.
In the following we consider the interval $[\vbar, \infty[$ containing the current version $\versionCurrent$ of the set.  Let $\Nbar = |\versionSet{\vbar}|$.
We allow up to $c \cdot \Nbar$ partial persistent updates during this interval of versions for a constant $0 < c < 1$, i.e., 
all versions~$v$, $\vbar \le v \le \versionCurrent$, satisfy $(1 - c) \cdot \Nbar \le |\versionSet{v}| \le (1 + c) \cdot \Nbar$.

Our data structure is built around four central parameters:
\[
\Delta = \lrCeil{\Beps} \quad\quad
H = 1 + \lrCeil{\log_{\Delta} \Nbar \, } \quad\quad
F = \lrCeil{ B^{1 - \epsilon} } \quad\quad
R = H \cdot 2\Delta \cdot F 
\]
The basic idea is to have a \BepsTree $T$ of degree at most $2\Delta$ (and degree at least $\Delta$, if only insertions can be performed, and degree at least 1 if deletions are allowed), where leaves (open rectangles) store between $4R$ and $10R$ updates (see Section~\ref{sec:updating_buffer_rectangles}) and all leaves are at the same layer. 
In Section~\ref{sec:constant_computation} we prove that $H$ is an upper bound on the height of~$T$ (number of nodes on a root-to-leaf path, excluding the leaves). 
Each internal node of $T$ will have a buffer of at most $2\Delta F = \OhTheta{B}$ updates yet to be applied to the leaves of the subtree rooted at the node. 
Note that $R$ is an upper bound on the total number of buffered updates along a root-to-leaf path in~$T$.
The essential property of the parameters is that 
$R/B= \OhTheta{H}=\OhTheta{\frac{1}{\epsilon}\logB \Nbar}$.

The geometric plane defined in Section~\ref{sec:persistence_geometry} 
is partitioned into a number of axis aligned rectangles~$[x, y[ \times [v, w[$, such that the number of updates in each rectangle is $\OhTheta{R}$.
For each rectangle we store a list of the updates in the rectangle in lexicographical order by first the value and secondly the version of that update. Note that this groups equal values consecutively in the list.
To allow efficiently locating a rectangle for a given version and value, we store a list indexed by version identifier, where we for each version~$v$ store a pointer to the root of a B-tree $P_v$ over the rectangles left-to-right containing~$\versionSet{v}$ (see Section~\ref{sec:updating_buffer_rectangles}).
Further, we require that each rectangle contains $\OhMega{R}$ values which are present in all versions the rectangle spans.
We denote such a value as \emph{spanning}. 
If $\Nbar = \Oh{R}$, all updates are stored in a single list. 

New updates are buffered, to achieve \IO efficient update bounds. The topmost rectangles, which cover the current version $\versionCurrent$, are all \emph{open}, with all other rectangles being \emph{closed}.
Crucially, new updates are always performed in the current version.
We maintain the invariant that for a buffered update, i.e., an update that has not yet been flushed to the corresponding rectangle, the corresponding rectangle must be open.

For the open rectangles, we store a \BepsTree $T$, such that recent updates to the open rectangles are buffered.
We let 
the maximum degree of an internal node in $T$ be $2 \Delta$.
Each internal node of $T$ contains a buffer of up to $2\Delta F$ updates, sorted lexicographically by value and version. Additionally, each update stores if the update is an insertion or deletion.
Consider a full buffer, i.e., it contains at least $2\Delta F$ updates, where each update should be \emph{flushed} to one of the at most $2\Delta$ children. Then, there must exist a subset of size at least $F$ updates, which should be flushed to the same child. 

The setup is illustrated on Figure~\ref{fig:geometric_partitioning}.
The vertical black and gray lines represent the version intervals containing a value.
The black lines represent updates present in the list of updates contained in that rectangle, while the gray lines and endpoints represent updates contained in buffers of $T$, which are illustrated at the top of the figure.

% --------------------------------------------------------------------
\subsection{Handling Queries and Updates}
\label{sec:query_update}
% --------------------------------------------------------------------

When performing $\Search(v, x)$, first the rectangle~$r$ covering point $(x, v)$ in the plane must be found. By using the B-tree $P_v$
associated with version $v$, $r$ can be found using $\Oh{H}$ \IOs.
If $r$ is closed, then all updates inside $r$ are contained in the sorted list of updates stored in $r$, and these can be scanned in $\Oh{R / B}$ \IOs.
If $r$ is open, then the result of the successor query may be affected by buffered updates, which are not stored in $r$. 
The rectangle~$r$ must therefore be \emph{actualized}, by merging all updates in buffers on the path from the root to $r$ with the updates in $r$.
The details of this operation is described in Section~\ref{sec:updating_buffer_rectangles}, where the actualize operation is shown to have an amortized $\Oh{H}$ \IOs.
After $r$ is actualized, the query continues by scanning the updates of $r$.
If the result of the successor query is not contained in the rectangle~$r$, then by the spanning requirement, the result of the query must be in the neighboring rectangle to the right, that similarly is actualized if it is open.
In total, the operation spends amortized $\Oh{H + R/B} = \Oh{\frac{1}{\epsilon} \logB \Nbar}$ \IOs. 
Note that the operation easily can be modified to support member, predecessor, and strict predecessor or successor queries. 

A $\Range(v, x, y)$ query is performed very similar to a \Search query. The query may however touch more than two rectangles.
Note that for the at most two rectangles containing the endpoints of the query we do not necessarily report all values they contain at version~$v$. These rectangles can be found using amortized $\Oh{H + R/B}$ \IOs by the argument above.
For each intermediate rectangle accessed (and possibly actualized if it is open), then by the spanning requirement, a constant fraction of the updates scanned result in reported values.
Since accessing a rectangle takes amortized $\Oh{H + R/B}$ \IOs, and each rectangle contains $\OhTheta{R} = \OhTheta{B \cdot H}$ values at version $v$, then amortized $\Oh{1 / B}$ \IOs are spent for each value reported for an intermediate rectangle.
In total, a $\Range$ query reporting $K$ values takes amortized $\Oh{H + R/B + K/B} = \Oh{\frac{1}{\epsilon} \logB \Nbar + K / B}$ \IOs.

Each update, either an \Insert or \Delete, is applied to the current version $\versionCurrent$ of the set.
The \BepsTree $T$ contains all buffered updates to the open rectangles, which cover the current version~$\versionCurrent$.
For an update operation, a tuple with the update and $\versionCurrent$ is added to the root buffer of $T$, which is stored in internal memory.
In Section~\ref{sec:updating_buffer_rectangles} it is shown that adding the update to the root buffer and handling possible \emph{buffer overflows} takes amortized $\Oh{H / F}=\Oh{\frac{1}{\epsilon B^{1-\epsilon}} \logB \Nbar}$ \IOs.

% --------------------------------------------------------------------
\subsection{Flushing Buffers}
\label{sec:updating_buffer_rectangles}
% --------------------------------------------------------------------

To argue about the amortized cost of flushing the content of buffers down the tree~$T$, we let the potential of each buffered update be $1 / F$ multiplied by the height of the buffer the update is stored in, with the root buffer being at the largest height. 
One unit of released potential can cover $\Oh{1}$ \IOs.
When adding an update to the tree, the root buffer is always stored in internal memory, and therefore no \IOs are needed to access it. 
However, the potential is increased by at most $1/F\cdot H$, and the operation therefore uses amortized $\Oh{H / F}$ \IOs.

\subparagraph*{Buffer Overflows.}

Each buffer at an internal node of $T$ contains at most $2\Delta F$ updates. If a buffer contains more than $2\Delta F$ updates, then a \emph{buffer overflow} is performed. 
Since each node of $T$ has at most $2 \Delta - 1$ children,
at least $F$ updates from the buffer must be to the same child.
These updates can be moved together to the buffer of that child.

A buffer overflow can happen in two cases. Either when an update is placed into the root buffer as the result of an update operation, or when updates are placed into a buffer because the parent buffer is overflowing.
Moving exactly $F$ updates out of a buffer, is always sufficient to make an overflowing buffer non-overflowing again.
A buffer overflow therefore only moves down a single path of $T$.

An overflowing buffer can be stored in $\Oh{1}$ blocks, since $2\Delta F + F=\Oh{B}$, and therefore the $F$ updates to remove can be found in $\Oh{1}$ \IOs.
If the overflowing updates are moved to a child buffer, these can be inserted via a merge in $\Oh{1}$ \IOs.
As $F$ updates are moved one layer down the tree, then the potential decreases by $1$, which is enough to cover the $\Oh{1}$ \IO cost of the overflow operation.

If the child is not an internal node of $T$, but an open rectangle,
then merging the $F$ overflowing updates into the list of updates in the rectangle takes $\Oh{R / B} = \Oh{H}$ \IOs.
As buffer overflows only move down a single path of the tree, then $\OhMega{H}$ overflows must have occurred before the overflow reaches the open rectangle. 
Merging the overflow into the list of updates in the open rectangle therefore does not increase the asymptotic amortized number of \IOs performed.

\subparagraph*{Actualizing.}

When \emph{actualizing} an open rectangle, all buffered updates to that open rectangle must be moved into the rectangle.
Note that all relevant updates are in the buffers on the root-to-leaf path in $T$ to the open rectangle.
Each of these buffers contains at most $2\Delta F=\Oh{B}$ buffered updates.
The total number of buffered updates on the path is at most $H \cdot 2 \Delta \cdot F = R$. 
For each node on the path from the root down to the rectangle the following is done.
Let $U$ be the updates on the path from the layers above in sorted order. Initially $U$ is empty.
To extend $U$ for each layer top-down, $U$ is merged with the relevant updates of the next buffer. This requires $\Oh{1 + |U| / B}$ \IOs, by scanning $U$ and the buffer. 
Since the $U$ updates are moved one layer down, they release potential $|U|/F \geq |U|/B$, 
that can cover $\OhTheta{|U|/B}$ \IOs, i.e., the amortized cost for actualizing one level of the tree is $\Oh{1}$ \IOs.
As there are $\Oh{H}$ layers of the tree, the at most $R$ relevant updates can be found in sorted order in amortized $\Oh{H}$ \IOs.
They can then be merged with the updates in the open rectangle in $\Oh{R / B}$ \IOs.
In total, the actualize operation requires amortized $\Oh{H + R / B} = \Oh{H}$ \IOs.

\subparagraph*{Finalizing.}

Each open rectangle is allowed to receive between $R$ and $2R$ updates before it is \emph{finalized}, converting it into a closed rectangle. 
When finalizing an open rectangle, all buffered updates to the rectangle are removed from $T$ and merged with the rectangle, to ensure that all buffered updates in $T$ are only to open rectangles.
We will now argue that open rectangles receive at most $2R$ updates in total by finalizing the rectangle as soon as $R$ updates have been added to it.
A finalize operation can be triggered from an actualize operation or from a buffer overflow.
Note that in both cases, the number of updates in the rectangle before the operation is at most $R - 1$.

An actualize operation may add at most $R$ buffered updates to a rectangle, i.e., at most $2R - 1$ total updates are placed in a finalized rectangle.
If the rectangle receives an update from a buffer overflow, then the overflow must have been triggered by an update in the root buffer. Buffered updates to add to the rectangle can only be the $R$ updates in buffers on the path, and the new update, which in total is at most $(R - 1) + R + 1 = 2R$ updates to add to the open rectangle.
Thus, by finalizing a rectangle as soon as it receives at least $R$ updates, it will contain between $R$ and $2R$ updates.

\subparagraph*{Spanning Requirement.}

We require that the first version of a rectangle contains $[4R, 8R[$ values. When finalizing the rectangle, $[R, 2R]$ updates have been performed and therefore at least $2R$ of the initial values are still present, that is, span all versions of the rectangle.
This ensures that the $\OhMega{R}$ spanning values requirement is met.

When finalizing a rectangle, it holds that $[2R, 10R[$ values are contained in the rectangle at version $\versionCurrent$.
New open rectangles must be created to span the value range of the closed rectangle, where the values present at version $\versionCurrent$ are contained.
If the count is in $[4R, 8R[$, then a single rectangle suffices.
If $[8R, 10R[$ values are present, then the range is \emph{split} in two rectangles at the median value, both with $[4R, 5R[$ values, and $T$ must be updated as described below.
Otherwise, version $\versionCurrent$ of the rectangle contains $[2R, 4R[$ values.
A sibling rectangle is finalized, to allow for a \emph{merge} of the rectangles to occur. 
Note that the early finalizing of the sibling preserves the $\OhMega{R}$ spanning values requirement of the sibling.
The combined present values is then $[4R, 14R[$. A split may need to be performed, i.e., the result is one or two new open rectangles.

\subparagraph*{Updating the $B^{\epsilon}$-Tree $T$.}

After finalizing open rectangles, the \BepsTree must be updated accordingly. 
If an open rectangle was split, then a new child is added to the parent node in~$T$ of the updated rectangle.
If this increases the degree of the node to $2 \Delta$, it is split by distributing its children into two new nodes, each with degree~$\Delta$. 
Its buffer is also partitioned so that each buffered update is placed in the buffer containing its relevant child. 
Splitting the node further introduces a new child in the parent of the node. 
Note that this may cascade up the tree, but only on the path towards the root.
If a merge of the rectangle occurred, then a child is deleted from the parent. 
When merging rectangles, the merged rectangles must be siblings in the tree.
The rectangle is merged with the left or right neighbor rectangle, 
which has a closest lowest common ancestor with the rectangle in~$T$, 
to ensures that the value range of existing nodes only increase.
This may cause the degree of nodes to be below~$\Delta$.
Notably, we do not merge internal nodes of~$T$, as this could create a large buffer that requires multiple flushes in different directions, known as \emph{flushing cascades}~\cite{BenderDFJK20}. 
We instead allow nodes to have a degree down to $1$, where deleting the last child results in deleting the path of consecutive degree-one from the child towards the root. 
As we show in Section~\ref{sec:constant_computation}, this does not affect the asymptotic height of the tree.

When finalizing a rectangle,
only the path from the finalized rectangle to the root may be affected. We therefore create the new tree~$T$ via \emph{path copying}, which preserves the old~$T$. 
The buffers of the copied nodes are moved, such that all buffers are present in the tree for the current version.
We maintain a list indexed by version identifier, that for any version~$v$ stores a pointer to the root of~$T$ for version $v$, i.e., the required tree~$P_v$. 

Updating the rectangles and the \BepsTree $T$ upon finalizing therefore requires scanning $\Oh{R}$ values and traversing a constant number of paths of length $\Oh{H}$ in~$T$, which takes $\Oh{R / B + H} = \Oh{R/B}$ \IOs.
As $\OhMega{R}$ updates must be applied to a rectangle before
finalizing it, this does not increase the asymptotic amortized cost of an update operation. Queries may also finalize rectangles, but already require amortized $\Oh{H}$ \IOs, causing no asymptotic query overhead. 

\subparagraph*{Space Usage.}

When finalizing a rectangle, $\OhMega{R}$ updates must have occurred in that rectangle. New rectangles are then created, which in total copies $\Oh{R}$ updates, and one path of the tree is copied. As the height of the tree is at most $H=\Oh{R / B}$, and the updates of the rectangles are stored in lists, the newly allocated space is $\Oh{R / B}$, which can be amortized over the $\OhMega{R}$ updates required for the finalization to happen.
In addition to the updates, initially $\Nbar$ values are stored across $\Oh{\Nbar / R}$ rectangles in lists, and a balanced initial \BepsTree is built on these initial rectangles, causing an initial space of $\Oh{\Nbar / B}$ blocks. 
As the structure allows for at most $c \cdot \Nbar$ updates, the space usage is therefore in total $\Oh{ \Nbar / B}$ blocks.

% --------------------------------------------------------------------
\subsection{Bounding the Tree Height}
\label{sec:constant_computation}
% --------------------------------------------------------------------

In this section we show that $H$ is an upper bound on the height of the \BepsTree $T$.

We define the \emph{weight}~$w_i$ of a node at height~$i$ in~$T$ to be the number of updates on values in the value range of the node.
The rectangles are at height $0$ of the tree, with the nodes of the tree starting at height $1$.
The updates are both the $\Nbar$ initial values as well as the up to $c \cdot \Nbar$ additional updates. 
It holds that the weight of a node is the sum of the weights of its children.
By induction on the number of updates we show that $w_i \ge B \Delta^i$ for all nodes at all heights, except for the root.
First note that the inequality holds for $i = 0$, as any rectangle contains at least $R \ge B$ updates when it was created.
Initially, $\Nbar$ updates are distributed into at most $\Nbar / R$ rectangles, where the number of updates in each rectangle is at least $R\geq B$.
Each internal node initially has degree~$[\Delta, 2 \Delta[$, except for the root that has degree~$[2, 2\Delta[$. 
By induction on the tree height~$i$, it holds that the initial tree satisfies~$w_i \ge B \Delta^i$, except for the root.
Each update affects some path of the tree. If the tree is not updated, then the weights of the nodes on the path can only grow, and therefore the inequality holds.
If rectangles are merged, then one rectangle disappears together with all the ancestors having only this single rectangle as a leaf. The other rectangle and its ancestors up to the least common ancestor of the two merged leaves get their value ranges expanded. It follows that the surviving nodes of a merge only can have their value range increase, and therefore the inequality holds.
If a split occurs in any node at height $i$, then the degree of the node before the split is~$2 \Delta$.
The node is split in two nodes at height $i$, each with $\Delta$ children. The weight of each of the two nodes is therefore at least $\Delta \cdot w_{i - 1} \ge \Delta \cdot B \Delta^{i - 1} = B \Delta^i$.
It therefore holds that~$w_i \ge B \Delta^i$.

Since the number of updates is at most $(1 + c) \cdot \Nbar$, we have $B \Delta^i \leq (1 + c) \cdot \Nbar$ for all nodes at height~$i$, except for the root. 
Since by definition $2 \leq \Delta \leq B$ and $c < 1$, we have $\Delta^{i+1} \leq 2\Nbar$, i.e., $i \le \log_{\Delta} \left( 2\Nbar \right) - 1 \leq \log_{\Delta} \Nbar$.
The height of $T$ is then at most the largest value of $i$ satisfying this inequality, plus one for the root, i.e., the height of~$T$ is at most $1+ \log_{\Delta} \Nbar \le 1 +\lrCeil{\log_{\Delta} \Nbar \, } = H$.

% --------------------------------------------------------------------
\subsection{Global Rebuilding}
\label{sec:global_rebuild}
% --------------------------------------------------------------------

The data structure above allows for an initial set of $\Nbar$ values to receive up to $c \cdot \Nbar$ persistent updates, for a constant $0 < c < 1$.
For any version~$v$, we have
$(1-c)\cdot\Nbar \leq N_v \leq (1+c)\cdot\Nbar$, i.e.,  $\Nv = \OhTheta{\Nbar}$.
Therefore, the asymptotic costs of all operations also hold with $\Nbar$ replaced by $\Nv$.

To allow for more than $c \cdot \Nbar$ updates, we create multiple copies of the data structure above using global rebuilding~\cite{Overmars83, OvermarsLeeuwen81}.
Whenever the current data structure reaches $c \cdot \Nbar$ updates, a new data structure is created with initial set~$\versionSet{\versionCurrent}$ and $\Nbar_{\mathrm{new}}=|\versionSet{\versionCurrent}|$ (and new $H$ and $R$ parameters),
with a new set of rectangles and a new \BepsTree $T$, where all buffers are empty.
We compute $\versionSet{\versionCurrent}$ by performing $\Range(\versionCurrent, -\infty, \infty)$ in 
amortized $\Oh{\frac{1}{\epsilon} \logB \Nbar + K / B}=\Oh{\Nbar/B}$ \IOs.
The new data structure can be build using $\Oh{\Nbar / B}$ \IOs by a single scan of the sorted list containing $\versionSet{\versionCurrent}$.

In the old data structure $c \cdot \Nbar$ updates have been performed before this rebuild is performed.
By amortizing the rebuild cost over these updates, the amortized cost of each update is increased by $\Oh{1/B}$ \IOs, i.e., the asymptotic amortized cost of an update is not increased.
As the space usage of the new data structure is $\Oh{\Nbar / B}$ blocks, a similar argument can be used to amortize the space usage over the updates, maintaining a linear space usage in the total number of updates.
This concludes the proof of Theorem~\ref{thm:main}.

% --------------------------------------------------------------------
\section{Worst-Case Bounds}
\label{sec:worst-case}
% --------------------------------------------------------------------

In this section, we describe how to achieve worst-case \IO guarantees instead of amortized under progressively weaker assumptions on the internal memory size $M$. 
Previous approaches to improving the amortized performance of ephemeral \BepsTrees, both in the randomized~\cite{BenderDFJK20} and worst-case~\cite{DasIaconoNekrich22} setting, assumed at least that $M = \OhMega{\frac{1}{\epsilon} B \logB \Nbar } = \OhMega{H B}$, which allows all buffers on a path to be sorted in internal memory, i.e., $\Oh{\sort{H B}} = \Oh{H}$ \IOs. 
First, in Section~\ref{sec:WC_large_memory}, we show that if $M = \OhMega{H B}$, our persistent structure can be deamortized without asymptotic overhead. 
Then, in Section~\ref{sec:WC_smaller_buffers_on_paths}, we relax the assumption to $M = \OhMega{B^{1-\epsilon}\log_2 \Nbar }$ using the \emph{subtracting game} studied by Dietz and Raman~\cite{DietzRaman93}. 
This represents a factor ${\Beps}/{\log_2 B}$ improvement on the assumption for the size of the internal memory. 
Finally, in Section~\ref{sec:WC_improving_range_queries}, we show worst-case results when only assuming $M \geq 2B$, which introduces a small additive overhead on all operations. 
We employ \emph{the zeroing} game by Dietz and Sleator~\cite[Theorem~5]{DietzSleator87} to avoid a multiplicative overhead for \Range queries. 

% --------------------------------------------------------------------
\subsection{Large Internal Memory Assumption}
\label{sec:WC_large_memory}
% --------------------------------------------------------------------

We first consider the case when $M = \OhMega{HB}$. 
When actualizing a rectangle (see Section~\ref{sec:query_update}), the buffers at the $\Oh{H}$ nodes along the root-to-leaf path, with a total size of $\Oh{H B}$, are merged to produce a sorted list of updates to apply to the rectangle. 
As shown in Section~\ref{sec:updating_buffer_rectangles}, this can be done in amortized $\Oh{H}$ \IOs, by merging the buffers top-down. 
In the worst case, this requires $\Oh{H^2}$ \IOs. 
By instead merging the buffers using an external memory sorting algorithm, the worst-case number of \IOs can be improved to $\Oh{\sort{H B}}$. 
Previous approaches to improving the amortized performance of \BepsTrees in the randomized~\cite{BenderDFJK20} and worst-case~\cite{DasIaconoNekrich22} settings both assumed at least that $M = \OhMega{H B}$, in which case the sorting term trivially disappears by performing the sorting internally after reading the $H$ buffers into internal memory. 
The remaining challenge was handling flushing cascades, which occur when merging internal nodes of $T$ results in large buffers requiring many flushes in different directions. 
For our structure, we avoid this issue by never merging internal nodes, and instead maintain the height of $T$ using global rebuilding. 
For the remainder of this section, we assume a large internal memory of size $M = \OhMega{H B}$ and describe how to achieve worst-case guarantees by incrementally performing amortized work. 

\subparagraph*{Queries.}

Finding and actualizing a relevant rectangle for a query takes $\Oh{H + \sort{H B}} = \Oh{H}$ \IOs when $M = \OhMega{H B}$. 
The worst case for a \Search and \Range query is therefore $\Oh{H}$ and $\Oh{\lrParents{1 + K / R} H} = \Oh{H + K / B}$ \IOs, respectively. 
Note that for a \Range query, for each rectangle that intersects the query, except for the leftmost and rightmost ones, $\OhMega{R}$ values are reported due to the $\OhMega{R}$ spanning values in each rectangle. 

\subparagraph*{Updates.}

When performing an update, it may be the $\lrCeil{c \cdot \Nbar \,}$'th update, which triggers a global rebuild of the structure based on a new $\Nbar$, which uses $\Oh{\Nbar / B}$ \IOs, as described in Section~\ref{sec:global_rebuild}. 
However, by performing the global rebuilding incrementally~\cite{Overmars83, OvermarsLeeuwen81} over the next $\OhTheta{\Nbar}$ updates, this does not increase the asymptotic worst-case number of \IOs of each update. 
While initializing the new structure there are still updates happening which must then be applied before it can take over.
By performing updates to the new structure at a sufficiently fast rate compared to the live structure this ensures that they stay within a constant factor of each other in size until the new structure takes over. 
Therefore, only the \IO cost of an update without global rebuilding needs to be considered.

Updates are inserted in the root buffer of the \BepsTree, as described in Section~\ref{sec:updating_buffer_rectangles}. 
By keeping the root buffer in internal memory this uses no \IOs. 
If the root buffer overflows, it may cause buffer overflows along a root-to-leaf path down to an open rectangle, which may then be finalized by performing an actualize operation followed by a path copy.
The update therefore requires $\Oh{H}$~\IOs in total under the large internal memory assumption.
However, each time the root buffer overflows, $F$ updates are removed from it, meaning this occurs at most every $F$th update.
Thus, when the root buffer overflows, we incrementally apply the update to the structure over the next $F$ updates, ensuring that $\Oh{H/ F}$ \IOs are performed per update in the worst case. 
To not interfere with the incremental work, we place new updates in a separate buffer while it is in progress and merge them with the root buffer when it is finished. 
If a path copy has occurred, the root pointers of the $F$ most recent versions must be updated to the new root. 
Since they are stored together in an array indexed by their version identifier this takes $\Oh{1}$ \IOs.
If a query occurs while an update is being performed incrementally, we complete the update before executing the query. 
This does not increase the asymptotic worst-case number of \IOs for queries. 

% --------------------------------------------------------------------
\subsection{Smaller Buffers on All Paths Using the Subtraction Game} \label{sec:WC_smaller_buffers_on_paths}
% --------------------------------------------------------------------

In the previous section, we described how as long as the internal memory can hold all values on a root-to-leaf path towards the same open rectangle, there is no overhead on worst-case queries and updates compared to the amortized bounds. 
To lower the possible number of such values, we will slightly change the flushing strategy described in Section~\ref{sec:updating_buffer_rectangles} where we only performed a flush when a buffer overflowed. 
Instead, for every $F$'th update, we flush along an entire root-to-leaf path, always flushing towards the child where most of the updates are going. 
We still flush at most $F$ values, which preserves the property that internal nodes of $T$ contain at most $2 \Delta F$ updates. 
In the following, we show that this flushing strategy guarantees that all buffers contain $\Oh{F \log_2 \Delta}$ updates going towards the same child, and therefore also the same leaf. 
This implies that the assumption $M = \OhMega{H F \log_2 \Delta} = \OhMega{B^{1-\epsilon} \log_2 \Nbar}$ is sufficient to achieve no overhead for worst-case queries and updates. 
This is a factor $\OhTheta{B^\epsilon / \log_2 B}$ improvement over the previous smallest assumption on $M$~\cite{DasIaconoNekrich22}.

We can view each node as playing the \emph{subtracting game} studied by Dietz and Raman~\cite{DietzRaman93} for the number of updates $x_1, x_2, x_3, ..., x_{2\Delta}$ going towards each of its at most $2\Delta$ children. 
In particular, when at most $F$ updates are flushed towards a node, if there are $\delta_i$ new updates going towards the $i$th child, then variable $x_i$ is increased by $\delta_i$. 
Then we flush towards the child $j$ where most of the updates are going which sets the variable $x_j = \max\!\lrCurlyBrackets{x_j - F, 0}$. 
Following~\cite[Theorem 3]{DietzRaman93}~and~Theorem~\ref{thm:subtraction_game_refined} in the Appendix, scaled by a factor of $F$, this guarantees $x_i = \Oh{F \log_2 \Delta}$ for any $i$. 

We also need to consider how merging and splitting in $T$ impacts the games. 
Only leaves of $T$, corresponding to open rectangles, are merged. 
When an internal node is split, it corresponds to evenly distributing the $x_i$ variables from one game to two new games, except for one variable that is split into two new variables, each with a smaller or equal value. 
When a leaf, i.e. an open rectangle, is merged or split, the one or two rectangles involved are first actualized, which sets their variables to zero, a stronger operation than subtracting. 
Thus, the variable for a new rectangle is always zero and variables on root-to-leaf paths to actualized rectangles may be decremented. 
In all cases and for all games, variables are either decremented without adding to the game, or a copy of an existing game is created, where all variables in the copy are equal or smaller in value than before. 
This concludes the proof of Theorem~\ref{thm:worst-case-large-cache}. 

% --------------------------------------------------------------------
\subsection{Improving Worst-Case \Range Queries Using the Zeroing Game} \label{sec:WC_improving_range_queries}
% --------------------------------------------------------------------

In this section, we consider the small-memory setting with $M \geq 2B$, to overcome the theoretical limitation of the memory assumptions made in Sections~\ref{sec:WC_large_memory}~and~\ref{sec:WC_smaller_buffers_on_paths}. 
Actualizing a rectangle by merging the relevant updates on a root-to-leaf path to a rectangle requires $\Oh{H + \gamma}$ total \IOs, where $\gamma = \sort{B^{1 - \epsilon} \log_2 \Nbar}$.
The construction from the previous section directly results in worst-case \Search queries and updates in $\Oh{H + \gamma}$ and $\Oh{\frac{1}{F} \lrParents{H + \gamma}}$ \IOs, respectively. 
However, since \Range queries are performed by repeatedly searching for the $\OhTheta{1 + K / R}$ rectangles intersecting the query, the worst-case number of \IOs is $\OhTheta{\lrParents{1 + K / R} \lrParents{H + \gamma}} = \OhTheta{H + \gamma + \frac{K}{B} \lrParents{1 + \frac{\epsilon \log_2 B}{B^\epsilon}\log_{M/B}\lrParents{B^{- \epsilon}\log_2\Nbar}}}$, notably with a multiplicative non-constant overhead on the reporting term. 
In this section, we describe how to guarantee \Range queries in worst-case $\OhTheta{H + \gamma + \frac{K}{B}}$ \IOs when $M \geq 2B$. 

The worst-case \IO cost of a \Range query can be improved by merging all the buffered updates to the open rectangles intersecting the query in a top-down, layer-by-layer fashion.
That is, by essentially actualizing all the open rectangles intersected by the \Range query simultaneously. 
We denote the updates already present in the rectangles by the \emph{partial output list}. 
Rather than applying the buffered updates to the rectangles, we merge them with the partial output list to obtain the final output. 
A given query $\Range(v, x, y)$ reports $\OhMega{R}$ values from each intermediate rectangles, i.e., all rectangles intersecting the query except for the two that contain the endpoints $x$ and $y$. 
Thus, in $\Oh{(1 + K / R) H + (R + K) / B} = \Oh{H + K / B}$ \IOs we can find all the relevant rectangles and compute the partial output list. 
To collect the buffered updates in~$T$ for the intermediate open rectangles in sorted order, we merge the updates down layer by layer. 
We only move down the updates to versions earlier or equal to $v$ since only these can affect the query result. 
Once obtained, these updates are merged with the partial output list using linear \IOs to report the output of the \Range query. 

The buffered updates are stored in $T$, which contains the open rectangles at version~$\versionCurrent$, however, the \Range query is on the rectangles present at version~$v$. 
Let $T_v$ denote the \BepsTree on open rectangles at version~$v$, i.e., the state of $T$ when version~$v$ was created. 
From $T_v$ to $T$ the tree may have changed, but no later updates are relevant for the query. 
Thus, the total number of relevant updates does not increase from $T_v$ to $T$, and each update remains on the root-to-leaf path towards the open rectangle to which the update is relevant.
The relevant updates in $T$ can be collected in sorted lists ordered by layer by traversing each root-to-leaf path in $T$ towards open rectangles intersecting the query using $\Oh{(1 + K / R) H}$ \IOs. 
To bound the \IOs to move the updates down layer by layer, we show that the total number of updates is $\Oh{B^{1 - \epsilon} \log_2\Nbar + K / H}$. 
To this end, we need the additional invariant that all nodes of degree one have empty buffers, which we show how to obtain below.
Let $T_v'$ be the subtree of $T_v$ consisting of all nodes on root-to-leaf paths to rectangles intersecting the query. 
Then split $T_v'$ into two root-to-leaf paths $p_x$ and $p_y$ to $x$ and $y$, respectively, along with all the subtrees hanging off $p_x$ or $p_y$. 
For a node on $p_x$ (symmetrically $p_y$) of degree at least two there may be one or more subtrees $T_{sub}$ hanging off the node. 
Since only the nodes of $T_{sub}$ with degree at least two have non-empty buffers, if $T_{sub}$ has $\ell$ leaves, the number of buffered updates in $T_{sub}$ is at most $\Oh{(\ell - 1) B}$.
Thus, the number of buffered updates in $T_v'$ excluding degree one nodes on $p_x$ and $p_y$ is $\Oh{K / H}$. 
A degree one node on $p_x$ and $p_y$ may have a large degree in $T_v$, but since it only has one child in the direction of the query, due to the subtracting game, it stores at most $\Oh{F \log_2 \Delta}$ relevant updates. 
The number of degree one nodes on $p_x$ and $p_y$ is at most $2H$, and they together contribute $\Oh{B^{1 - \epsilon} \log_2 \Nbar}$ buffered updates, which we locate and sort separately using $\Oh{H + \gamma}$ \IOs.
For the remaining $\Oh{K / H}$ buffered updates, we merge them layer by layer using $\Oh{H + K/B}$ \IOs. 
In total, the worst-case number of \IOs to perform a \Range query is $\Oh{H + \gamma + K/B}$ \IOs. 

\subparagraph*{Empty Buffers for Degree one Nodes. }
To ensure that each node of degree one has an empty buffer, we alter the buffer capacity of nodes to scale with the degree.
Let the capacity of the buffer of a node with degree $d \ge 2$ be at most $F \cdot \min\{ 2 d, 2 \Delta \}$, with nodes of degree one having a buffer capacity of $0$.
When flushing a node, as the maximum number of updates in the buffer scales with the degree, then at least $F$ updates going to the same child can be found when overflowing. 
When splitting a node, it must have degree $2\Delta$, resulting in the two new nodes having degree~$\Delta$, which therefore does not decrease the buffer capacity, and flushing is not needed.
When a child of a node is removed due to merging rectangles, the degree of the node is decreased by one. 
If the degree remains at least two, at most two flushes are required to get the buffer capacity within bounds. 
Otherwise, if the degree drops from two to one, at most four flushes are needed.
To avoid cascading merges of rectangles, we do not finalize a rectangle once it receives a certain number of updates.
Instead, we finalize the rectangle that has received the most updates, provided it has received at least $R$ updates. 
This last condition ensures a bound on the space usage.
Including the initial flush from the root buffer, then when a rectangle is finalized, at most $5F$ updates have been flushed into open rectangles.

Let $U_i$ denote the number of updates to the $i$th rectangle, excluding the initial insertions.
We extend the data structure to include an array over all open rectangles, where index $j$ stores a blocked-linked-list of all rectangles where the number of updates is $U_i = j$.
Each rectangle has a double linked pointer between its location in the array of lists and the rectangle. This allows for moving a rectangle to a new entry in the array, when it receives updates, as well as finding a rectangle which have received the most updates, by scanning the list. 

To show that $U_i$ is bounded by $\Oh{R}$, we apply the \emph{zeroing game} of Dietz and Sleator~\cite{DietzSleator87}, using the variables
$x_i = \max \left\{0, \frac{U_i - R}{5F} \right\}$ if rectangle $i$ is open and $x_i = 0$ if it is closed.
For open rectangles, $x_i$ count the number of units of $5F$ updates received beyond the first $R$ updates.
This ensures that the variables are incremented by at most $1$ in total for each round, when at most $5$ flushes of size at most $F$ are flushed into the open rectangles.
When finalizing a rectangle it becomes closed, which ensures that $x_i = 0$, matching the zeroing step.
We bound the total number of rectangles by $\Nbar$ and therefore also the number of variables.
Following \cite[Theorem~5]{DietzSleator87} and a proof similar to Theorem~\ref{thm:subtraction_game_refined} in the Appendix, we have that for any $i$ then $x_i \le \log_2 \Nbar + 1$ at any time. 
Consequently, it follows that $U_i \le 5F (\log_2 \Nbar + 1) + R$. 
It can be shown that $5 F (\log_2 \Nbar + 1) \le 2 R$, when $\Nbar \ge 8$, by unfolding $R$ and simplifying the inequality to show that $\frac{\epsilon \log_2 B}{B^\epsilon} \lrParents{1 + \frac{1}{\log_2 \Nbar}} \le \frac{4}{5}$.
It therefore holds that each rectangle receives at most $U_i \le 3 R$ updates, due to the zeroing game.
Thus, when finalizing a rectangle, at least $R$ updates have been performed. 
Including the updates from the buffers on the path towards the rectangle, the total number of updates applied is between $R$ and $4R$.
By ensuring that each rectangle contains $[8R, 16R[$ initial values, the rebalancing operations are possible, and the spanning criteria remains satisfied.

An update therefore performs at most $5$ flushes using $\Oh{H}$ \IOs, along with locating and finalizing a single rectangle in respectively $\Oh{R / B} = \Oh{H}$ and $\Oh{H + \gamma}$ \IOs.
Performing this operation incrementally allows for updates to spend worst-case $\Oh{\frac{1}{F} \lrParents{H + \gamma}}$ \IOs.
If a query happens while an incremental update is being performed, the incremental update is completed, using at most $\Oh{H + \gamma}$ \IOs, which does not increase the total cost of the query.
When a \Search query happens, similar to the new \Range query, we do not apply the relevant updates on the path to the open rectangle to avoid queries interfering with the zeroing game.
This concludes the proof of Theorem~\ref{thm:worst-case-small-cache}. 

% --------------------------------------------------------------------
\section{Discussion and Open Problems}
% --------------------------------------------------------------------

Global rebuilding, as described in Section~\ref{sec:global_rebuild}, allows for constructing a partially-persistent set of any sorted set in a linear number of \IOs,  without creating the set anew by a sequence of insertions.
Symmetrically, it is possible to \emph{purge} all versions of the set older than some threshold, without performing all updates anew. This problem was first motivated by Becker, Gschwind, Ohler, Seeger, and Widmayer~\cite{BeckerGOSW96}.
As our data structure consists of multiple independent data structures covering disjoint version intervals, then all data structures which only cover versions to be purged can be removed efficiently by a linear number of \IOs.
For both use cases, the space usage is asymptotically linear in the size of the oldest stored set and the number of updates performed.

Further, global rebuilding allows for a crude fully persistent data structure, which supports efficient buffered updates and queries, but where cloning past versions requires a linear number of \IOs.
The fully persistent data structure by Brodal, Rysgaard, and Svenning~\cite{BrodalRysgaardSvenning23} allows cloning past versions in worst case $\Oh{1}$ \IOs. They do, however, not buffer updates, which therefore are amortized and a factor $\Oh{1 / B^{1- \epsilon}}$ slower than our data structure.
Our data structure is therefore better when there are many updates, but few clones of past versions happening.
Further, our data structure is simpler. 
It remains an open problem to design buffered fully-persistent search-trees, which remains efficient for clone operations.

In Section~\ref{sec:worst-case} we showed how to achieve worst-case bounds matching those of ephemeral \BepsTrees, when $M = \OhMega{B^{1-\epsilon} \log_2 N}$. This is an improvement by a factor $\OhTheta{B^\epsilon / \log_2 B}$ on the required lower bound on $M$ over the worst-case results of \BepsTrees by Das, Iacono, and Nekrich~\cite{DasIaconoNekrich22}.
It remains an open problem to show a worst-case \IO lower bound dependency on $M$ or to find a structure with worst-case \IO guarantees matching the amortized \IO bounds for $M=2B$.

% --------------------------------------------------------------------
\bibliographystyle{plainurl}
\bibliography{references}
% --------------------------------------------------------------------

\appendix

% --------------------------------------------------------------------
\section{Subtraction Game}
% --------------------------------------------------------------------

Dietz and Raman~\cite{DietzRaman93} considered the following \emph{subtraction game}. The game is played on $n$ non-negative real variables $x_1,\ldots,x_n$, 
initially all zero. The game progresses in \emph{rounds}. Each round consists of 
an \emph{increment step} followed by a \emph{subtraction step}.
In the increment step an adversary selects $n$ non-negative real values $\delta_1,\ldots,\delta_n$, where $\sum_{i=1}^n \delta_i \leq 1$, and sets $x_i \leftarrow x_i+ \delta_i$.
In the subtraction step a largest $x_i$ is decremented by setting $x_i \leftarrow \max\{0, x_i - 1\}$.
Dietz and Raman {\cite[Theorem 3]{DietzRaman93}} proved an upper bound on all variables of $1+\ln n$.
Below we improve this bound to be $H_{n-1}\leq 1+\ln (n-1)$. Note that $H_{n-1} \leq \log_2 n$ for $n \geq 1$.
The proof extends to the \emph{zeroing game} of Dietz and Sleator~\cite[Theorem~5]{DietzSleator87}, where the subtraction step is replaced by a \emph{zeroing step}, in which a largest $x_i$ is set to $0$.
For both games, it holds that for all $i$, $x_i < H_{n-1} + 1$ after each step.

\begin{theorem}
\label{thm:subtraction_game_refined}
    The subtraction game on $n \geq 2$ variables guarantees all $x_i < H_{n-1} \leq 1 + \ln (n-1)$ after each round.
\end{theorem}
\begin{proof}
    For $1\leq k\leq n$, let $s_k$ denote a strict upper bound on the sum of the $k$ largest variables in the game after any round. The goal is to find a small valid~$s_1$. Below we argue that 
    $s_n = n-1$ and $s_k = \frac{k}{k+1} (1+s_{k+1})$, for $1 \leq k < n$, are valid upper bounds. By induction for decreasing $k$, we have $s_k \geq k$ for $1 \leq k < n$: $s_{n-1} = \frac{n-1}{n}(1+s_n)=\frac{n-1}{n}(1+n-1)=n-1$ and $s_k=\frac{k}{k+1}(1+s_{k+1})\geq \frac{k}{k+1}(1+k+1) > k$ for $1\leq k < n-1$.
    By induction for increasing~$k$, we have
    $s_1=\sum_{j=2}^{k} \frac{1}{j} + \frac{1}{k} s_k$ for $2 \leq k \leq n$, i.e., 
    $s_1=\sum_{j=2}^{n} \frac{1}{j} + \frac{n-1}{n}=\sum_{j=1}^{n-1} \frac{1}{j}=H_{n-1}$.
    
    Initially, all $x_i$ are zero, i.e., all $s_1,\ldots,s_n$ are valid strict upper bounds. By induction on the number of rounds, we show that the upper bounds remain valid after each round.
    
    Consider $s_n$, and assume there exists a round where the total sum is $< n-1$ before the round and $\geq n-1$ after the round. Since the total sum increases, the subtraction step must have decreased the total sum by $<1$. It follows that all variables after the increment step must be $<1$, and the subtraction step sets one variable to zero, i.e., the total sum after the increment step is $< n-1$. This contradicts the assumption that the total sum is $\geq n-1$ after the round, i.e., the total sum after each round is $< n-1$.

    Next, consider $s_k$, for $1 \leq k < n$. We first consider the case where the subtracted variable $x_i < 1$ before the subtraction step, i.e., all variables are $<1$. Then, after the subtraction step the sum of the $k$ largest variables is $<k\leq s_k$. 
    Otherwise, $x_i \geq 1$ after the increment step.
    If the variable $x_i$ is among the $k$ largest variables after the round, then the sum of the $k$ largest variables can increase by at most $\sum_{j=1}^{n} \delta_j - 1\leq 0$ in the round
    (the sum of the $k$ largest variables after the round consists of the same variables as before the round, or variables with smaller value before the round), i.e., the sum of the $k$ largest variables does not increase by the round. 
    Otherwise, $x_i$ is not among the $k$ largest variables after the round, but $x_i$ is the largest after the increment step, where the $k+1$ largest variables have sum $<\sum_{j=1}^{n} \delta_j+s_{k+1}$. After subtracting~$x_i$ (i.e., after the round), the sum of the new $k$ largest variables is
    $<\frac{k}{k+1}\big(\sum_{j=1}^{n} \delta_j+s_{k+1}\big)\leq \frac{k}{k+1}(1+s_{k+1})= s_k$.
\end{proof}

That the analysis is tight follows from the following strategy, as also described in the PhD thesis of Raman~\cite[Section 2.2.3]{RamanPhD}. We first perform sufficiently many initial rounds, where after $r$ rounds one variable has value zero, say $x_1=0$, and all other $n-1$ variables have value $\epsilon_r = 1-\big(1-\frac{1}{n}\big)^r$. Note $\epsilon_0=0$ and $\epsilon_r \to 1^-   $ for $r \to \infty$. In round~$r$ we let $\delta_1=\epsilon_{r-1} + (1 - \epsilon_{r-1})/n$ and $\delta_2=\cdots=\delta_n = (1- \epsilon_{r-1})/n$. This ensures that the increment step makes all $n$ values have value $\epsilon_r=\epsilon_{r-1} + (1 - \epsilon_{r-1})/n$ before the subtraction step. By induction it follows that $\epsilon_r = 1-\big(1-\frac{1}{n}\big)^r$.
By performing a sufficient number of initial rounds, $n-1$ variables can achieve value $1-\epsilon$ arbitrary close to one. In the next $n-2$ rounds, for $j=n-1,\ldots,2$, we distribute value $1/j$ to the $j$ variables with maximum value $1-\epsilon + \sum_{i=j+1}^{n-1} \frac{1}{i}$.
The final maximum value is $1-\epsilon + \sum_{i=2}^{n-1} \frac{1}{i}=H_{n-1} - \epsilon$.

\end{document}